\documentclass[12pt]{article}
\usepackage{latexsym}
\usepackage{amsmath}
\usepackage{amssymb}
\usepackage{amscd}
\usepackage{graphicx}
\newtheorem{theorem}{Theorem}[section]
\newtheorem{lemma}[theorem]{Lemma}
\newtheorem{corollary}[theorem]{Corollary}
\newtheorem{proposition}[theorem]{Proposition}
\newtheorem{definition}[theorem]{Definition}
\newtheorem{remark}[theorem]{Remark}
\newcommand{\EQ}{\begin{equation}}
\newcommand{\EN}{\end{equation}}
\newenvironment{proof}{\indent{\em Proof.  }}{\hfill{$\Box$\bigskip\newline}}

\newcommand{\rank}{\mbox{\rm rank}}

\newcommand{\wt}{\mbox{\rm wt}}

\newcommand{\ba}{{\bf a}}
\newcommand{\bb}{{\bf b}}

\newcommand{\bc}{{\bf c}}

\newcommand{\bx}{{\bf x}}

\newcommand{\bv}{{\bf v}}
\newcommand{\bu}{{\bf u}}
\newcommand{\bw}{{\bf w}}
\newcommand{\bo}{{\bf 0}}

\newcommand{\F}{\mathbb{F}}

\newcommand{\dd}{\displaystyle}

\newcommand{\syn}[1]{S_{\textbf{#1}}}
\newcommand{\pcm}[1]{H_{(#1)}}

\title{On lifting perfect codes\footnote{This work has been partially supported by the Spanish MICINN grants MTM2009-08435; PCI2006-A7-0616; the catalan grant 2009SGR1224 and also by the
Russian fund of fundamental researches 09 -01 - 00536.}}

\author{J. Rif\`{a} \\
Department of Information and Communications Engineering,\\
 Universitat Aut\`{o}noma de Barcelona,\\
  08193-Bellaterra, Spain\\
 (email:~josep.rifa@autonoma.edu)\\
\and V. A. Zinoviev\\Institute for Problems of Information
Transmission,\\ Russian Academy of Sciences,\\ Bol'shoi Karetnyi
per. 19, GSP-4, Moscow,\\ 127994, Russia (e-mail:\, zinov@iitp.ru).}

\begin{document}

\maketitle

\begin{abstract}
In this paper we consider completely regular codes, obtained from perfect (Hamming) codes by lifting the ground field. More
exactly, for a given perfect code $C$ of length $n=(q^m-1)/(q-1)$
over $\F_q$ with a parity check matrix $H_m$, we define a new code
$C_{(m,r)}$ of length $n$ over $\F_{q^r}$,~$r \geq 2$, with this
parity check matrix $H_m$. The resulting code $C_{(m,r)}$ is completely
regular with covering radius $\rho = \min\{r,m\}$.  We compute the
intersection numbers of such codes and, finally, we prove that Hamming
codes are the only codes that, after lifting the ground field, result in completely regular
codes.
\end{abstract}


\section{Introduction and Background}\label{int}

Let $\F_q=GF(q)$ be a Galois field with $q$ elements, where $q$ is a
prime power.  A $q$-ary linear {\em code} $C$ of length $n$ is a subset of $\F_q^{\,n}$. If
$C$ is a $k$-dimensional linear subspace of $\F^{\,n}_q$, then $C$ is a
$q$-ary {\em linear} code, denoted by $[n,k,d]_q$, where $d$ is the
{\em minimum distance} between any pair of codewords (the distance, or the Hamming distance, between two vectors is the number of coordinates they disagree). Given $\bv\in\F^n_q$, denote by $supp(\bv)$ the {\em support} of $\bv$, that is, the set of its nonzero coordinate positions. For $a \in \F_q$
and a given $C$ denote $a\,C = \{a\,\bc: ~\bc \in C\}$. For two codes
$A$ and $B$ of the same length let $A + B$ be its direct product, i.e.
$A + B = \{\ba + \bb:~\ba \in A,~\bb \in B\}$.

A matrix $M$ is called monomial if there is exactly one nonzero entry
in each row and column.
Let $C$ be a linear code of length $n$ over $\F_q$, a finite field of $q$ elements, where $q$ is a prime power. The automorphism
group $Aut(C)$ of $C$ usually consists of all
$n\times n$ monomial matrices $M$ over $\F_q$ such that $\bc M \in C$
for all $\bc \in C$. In those cases when $q$ is a power of a prime number, then $Aut(C)$ also
contains all the field automorphisms of $\F_q$ which preserve $C$. Note that, for
binary codes, the automorphism group $Aut(C)$ coincides with the subgroup of the
symmetric group $S_n$ consisting
of all $n!$ permutations of the $n$ coordinate positions which send $C$ into itself.

For a $q$-ary code $C$ with minimum distance $d$ denote by
$e = \lfloor (d-1)/2 \rfloor$ its {\em packing radius}.
Given any vector $\bv \in \F^n_q$, its {\em distance} to the code
$C$ is
\[
d(\bv,C)=\min_{\bx \in C}\{ d(\bv, \bx)\}
\]
and the {\em covering radius} of the code $C$ is
\[
\rho=\max_{\bv \in \F^n_q} \{d(\bv, C)\}.
\]
Clearly $e\leq\rho$ and $C$ is a {\em perfect} code, if
$e=\rho$.

For a given $q$-ary code $C$ with covering radius
$\rho=\rho(C)$ define
\[
C(i)~=~\{\bx \in \F_q^{\,n}:\;d(\bx,C)=i\},\;\;i=1,2,...,\rho.
\]

\begin{definition}\label{de:1.1}  A code $C$ is completely regular,
if for all $l\geq 0$ every vector $x \in C(l)$ has the same
number $c_l$ of neighbors in $C(l-1)$ and the same number
$b_l$ of neighbors in $C(l+1)$. Also, define
$a_l = (q-1){\cdot}n-b_l-c_l$ and note that $c_0=b_\rho=0$.
Define the intersection  array of $C$ as $(b_0, \ldots, b_{\rho-1}; c_1,\ldots, c_{\rho})$.
\end{definition}

The group $Aut(C)$ induces an action on the set of cosets of $C$ in the
following way: for all $\phi\in Aut(C)$ and for every vector
$\bv \in \F_q^{\,n}$ we have $\phi(\bv + C) = \phi(\bv) + C$.

In~\cite{Sole} it was introduced the concept of completely transitive
binary linear code and in~\cite{giudici} there is a generalization leading to the concept of coset-completely transitive code.

Let $C$ be a linear code over $\F_q$ with covering radius $\rho$, then
$C$ is coset-completely
transitive if $Aut(C)$ has $\rho +1$ orbits when acting on the cosets of $C$.
Since two cosets in the same orbit have the same weight distribution, it is clear that
any coset-completely transitive code is completely regular. Obviously, it may exist completely regular codes which are not coset-completely transitive.

To illustrate this last point we will show an easy example. Take the repetition code $C$ of length three over the field $\F_q$, where $q=2^4$. The code $C$ has 16 elements and it is not coset-completely transitive, but completely regular with intersection array  $( 45, 28; 1, 6 )$. Clearly, the automorphism group $Aut(C)$ as defined before has order $|Aut(C)|=15\cdot6\cdot4=360$ and if all the cosets in $C(2)$ are in the same orbit of $Aut(C)$ it should be that the number of cosets in $C(2)$ is a divisor of $360$. But, this is not true, since the number of cosets in $C(2)$ is $210$ (indeed, $210=16^2-1-15\cdot 3$).

It has been conjectured for a long time that if $C$ is a
completely regular code and $|C|>2$, then $e \leq 3$ \cite{Neum}.
Hence, the existing completely regular codes have a small error correcting
capability. With respect to the covering radius, Sol\'{e} in
\cite{Sole} uses the direct sum of $\ell$ copies of fixed perfect
binary codes of length $n$, with $e=1$, to construct infinite families of
binary completely regular codes of length $n{\cdot} \ell$ with
covering radius $\rho =\ell$. Thus, the
covering radius of the resulting code is growing linearly to infinity with
the length of the code. In~\cite{Rif2} it was described a method of constructing of linear completely regular codes with arbitrary covering radius, which is constant when the length is growing. Now, in the present paper we describe a new method, based on lifting perfect (Hamming) $q$-ary codes up to a new finite field over $q^r$. This method allow us to construct new completely regular codes with arbitrary covering radius, growing to the infinity with the length of the code, but not in a linear way like in~\cite{Sole} but in a logarithmic way. We compute the intersection array for all these new completely regular codes. Finally, we show that the Hamming codes are the only codes that, after lift the ground field, give completely regular codes

\section{Lifting perfect $q$-ary codes}

Let $C$ be a $q$-ary linear nontrivial code of length $n$ over the ground field $\F_q$ (hence, $1 < |C| < q^{n-1}$) with minimum distance $d\geq 3$. Let $H$ be a parity check matrix of $C$ and consider the code $C_{r}$ of length $n$ over the field $\F_{q^r}$ with the same parity check matrix $H$. The code $C_{r}$ is a well defined code, since the values in $H$ are from $\F_q \subset \F_{q^r}$.
We will say that code $C_r$ is obtained by lifting $C$.

The next Proposition shows a representation for the code obtained by lifting.

\begin{proposition}\label{rep}
Let $C_r$ be the code of length $n$ over the field $\F_{q^r}$, obtained by lifting the code $C$ over $\F_q$. Then:
$$C_r = C + \alpha C + \alpha^2 C + \cdots  + \alpha^{r-1} C,$$ where $\alpha$ is a primitive element in $\F_{q^r}$.
\end{proposition}
\begin{proof}
By construction, it is easy to see that $C_r \subset C + \alpha C + \alpha^2 C + \cdots + \alpha^{r-1} C$. The statement follows by computing the cardinality of $C_r$ and $C + \alpha C + \alpha^2 C + \cdots + \alpha^{r-1} C$.
\end{proof}

From~the previous Proposition it follows immediately that the minimum distance of the lifted code $C_{r}$ is the same that in code $C$ over $\F_q$. We can specify a little more in the following Lemma.

\begin{lemma}\label{lem:2.9}
Let $C_r$ be the code of length $n$ over the field $\F_{q^r}$, obtained by lifting the code $C$ over $\F_q$. Let $\bc$ be a codeword of $C_r$ of minimum weight. Then $\bc = \beta\,\bc'$ where $\beta \in \F_{q^r}$ and $\bc'$ is a codeword of minimum weight in $C$.
\end{lemma}

\begin{proof}
From~Proposition~\ref{rep} is easy to see that the minimum distance in $C_r$ is less or equal to the minimum distance in $C$. Vice versa, the argument is reduced to show that there are not linearly independent minimum weight vectors in $C$ with the same support. Hence, take two independent minimum weight vectors $\bv, \bu \in C$ with the same support. Now, notice that it is possible to compute the vector $a\bv + b\bu$ with appropriate values $a, b \in \F_q$ such that the support of the new vector $a\bv + b\bu$ is strictly less. Hence, the weight of  $a\bv + b\bu \in C$ is less than the minimum weight of $C$, which is contradictory.
\end{proof}

Denote by $\pcm{m,q}$ the parity check matrix of a perfect Hamming
$[n,k,3]_q$-code over $\F_q$, where $n=(q^m-1)/(q-1)$.
Let $\{\xi_0=0,\xi_1=1,\ldots,\xi_{q-1}\}$ denote the elements of
$\F_q$. Then the matrix $\pcm{m,q}$  can be expressed, up to equivalence,
through the matrix $\pcm{m-1,q}$ as follows \cite{Sem3}:
\[
 \pcm{m,q} = \left[
\begin{array}{cc|c|c|c|c}
&\,1\,&\,0\cdots 0\,&\,1\cdots 1\,&\,\cdots\,&\,\xi_{q-1}\cdots
 \xi_{q-1}\,\\\hline
&\,{\bf 0}\, &\pcm{m-1,q}\,&\,\pcm{m-1,q}\,&\,\cdots \,&\,\pcm{m-1,q}\,\\
\end{array}\right],
\]
where ${\bf 0}$ is the zero column and where $\pcm{1,q} = [1]$.

Let $C$ be the Hamming code over $\F_q$ with parity check matrix
$\pcm{m,q}$. Denote by $C_{(m,r)}$ the lifting $C$,
i.e. $C_{(m,r)}$ is the code over $\F_{q^r}^n$ with parity check matrix $\pcm{m,q}$.

First we find the covering radius for $C_{(m,r)}$.
Take any vector $\bv$ in the ambient space of $C_{(m,r)}$, so $\bv \in \F_{q^r}^{\,n}$ and compute the syndrome $\bv \pcm{m,q}^T \in \F_{q^r}^m$. We can represent this syndrome $\bv \pcm{m,q}^T$ by a $r\times m$ matrix over $F_q$, taking as columns the representation of the elements in $\F_{q^r}$ using a fixed basis in $\F_q$. Call $S_{\bv}$ this matrix and denote by $\rank(S_{\bv})$ its rank over $F_q$.

\begin{lemma}\label{lem:2.1}
Let $C_{(m,r)}$ be the code with parity check matrix $\pcm{m,q}$. Then, for any vector $\bv \in \F_{q^r}^{\,n}$ with weight $\wt(\bv)$ we have that $\rank(S_{\bv})\leq \wt(\bv)$.
\end{lemma}

\begin{proof}
Assume that $\bv$ is a nonzero vector and also that all the coordinates of $\bv$ belong to the ground field $\F_q$. In this case, it is clear that $\syn{v}$ has rank $1$. In the general case, as $\pcm{m,q}$ is the parity check matrix of a Hamming code over the ground field $\F_q$, we find that the rank of $\syn{v}$ depends on the number of independent coordinates of  the vector $\bv$ over $\F_q$. When $r \leq m$ we have that this number $r$ is exactly the rank and, when $r>m$ the rank is $m$. In any case, $\rank(S_{\bv})= min\{r,m\} \leq \wt(\bv)$.
\end{proof}

\begin{proposition}\label{dist}
For any vector $\bv \in \F_{q^r}^{\,n}$ we have that $d(\bv,C_{(m,r)})=\rank(\syn{v})$.
\end{proposition}

\begin{proof}
Matrix $S_{\bv}$ can be transformed into a diagonal matrix $P$ with only $l=\rank(S_{\bv})$ nonzero diagonal values. The transformation can be carried out by multiplying by the corresponding nonsingular matrices over $\F_q$. Hence, we can write:
\begin{equation}\label{eq:p}
P=A\,S_{\bv}B =A\,v\pcm{m,q}^TB,
\end{equation}
where $A$ is a $r\times r$ nonsingular matrix over $\F_q$ and $B$ is a $m\times m$ nonsingular matrix over $\F_q$.

We know that a Hamming parity check matrix is unique, up to permutations and scalar multiplications of columns. Hence, matrix $\pcm{m,q}^TB$ in~(\ref{eq:p}) represents a permutation of the $n$ coordinates, followed by scalar multiplications. Matrix $A$ makes a linear transformation in the vector space $\F_q^r$ over the ground field $\F_q$ (as we already mentioned, we use a representation of elements of $\F_{q^r}$ as elements of $\F_q^r$ by using a fixed basis). It is easy to see that the weight of the vectors $\bv \in \F_{q^r}^n$ is not changed after doing these above transformations. Looking at the matrix $P$ in~(\ref{eq:p}), it is easy to find a vector $\bw \in \F_{q^r}^n$ of weight $l=\rank(\syn{v})$ with syndrome equal to $P$, so $\syn{w}=P$. From Lemma~\ref{lem:2.1} we conclude that $d(\bv,C_{(m,r)})=\rank(\syn{v})$.\end{proof}

\begin{corollary}\label{coro:2.3}
The covering radius of code $C_{(m,r)}$ is $\rho=min\{r,m\}$.
\end{corollary}

\begin{proof}
It is straightforward to see that $\rho=min\{r,m\}$, since the distance of any vector $\bv \in \F_{q^r}^n$ to $C_{(m,r)}$ is given by $rank(S_{\bv})$ and $S_{\bv}$ is a $r\times m$ matrix over $F_q$.
\end{proof}

The second important goal is to prove the completely regularity for these codes over $\F_{q^r}$, which have $\pcm{m,q}$ as a parity check matrix.

To prove the condition of completely regularity we start by slightly modifying the usual definition of completely transitive codes~\cite{Sole,giudici}, and
introducing the concept of $(r,q)$-completely transitivity, which will be helpful throughout this paper. But, first, we give a different and convenient definition of the automorphism group of a code $C_r$ over $\F_{q^r}$, which we will denote by $Aut_{(r,q)}(C)$.

\begin{definition}
Define the automorphism group $Aut_{(r,q)}(C_r)$ of a code $C_r$ over $\F_{q^r}$ as all the
$n\times n$ monomial matrices $M$ over $\F_{q^r}$ such that $\bc M \in C_r$
for any $\bc \in C_r$; all the field automorphisms of $\F_{q^r}$ which preserve $C$ and, moreover, all the vector space morphisms of $\F_{q^r}$ over $\F_q$ which preserve $C_r$.
\end{definition}

Finally, we define the concept of $(r,q)$-completely transitivity.

\begin{definition}
Let $C_r$ be a $[n,k,d]_{q^r}$ code with covering radius $\rho$. We will say that $C_r$ is $(r,q)$-completely transitive if $Aut_{(r,q)}(C_r)$ has $\rho +1$ orbits when acting on the cosets of $C_r$.
\end{definition}

\begin{proposition}\label{prop:1.2}
If $C_r$ is a $(r,q)$-completely transitive code, then $C_r$ is completely regular.
\end{proposition}

\begin{proof}
It is enough to show that all the cosets in the same orbit have the same weight distribution. It is straightforward to see this, since an automorphism of $Aut_{(r,q)}(C_r)$ acts over the vectors by permuting coordinates, multiplying by a nonzero scalar and linearly modifying a coordinate. Indeed, each coordinate is an element of $\F_{q^r}$ and could be seen as a vector of $r$ coordinates in $\F_q$; after the modification we obtain another coordinate consisting in a change of basis in $\F_q^{\,r}$ as a vector space over $\F_q$. In all the cases a nonzero coordinate always gives a nonzero coordinate and the weight of the initial vector is maintained.
\end{proof}

Going again to the example we introduced in Section~\ref{int}, we see that now
$$|Aut_{(4,2)}(C)| = 360\cdot 20160 = 7257600,$$ where $20160=15\cdot 14\cdot 12\cdot 8$ comes from the number of vector space morphisms of $\F_{2^4}$ over $\F_2$ (or, equivalently, from the order of the general linear group $GL_4(\F_2)$ of degree $4$ over $\F_2$). As we will see, applying next Theorem, the code in that example is $(r,q)$-completely transitive, and so, completely regular.

\begin{theorem}\label{theo:2.1}
Let $C_{(m,r)}$ be the code of length $n$ over the field $\F_{q^r}$ with parity check matrix $\pcm{m,q}$, where $n=\frac{q^m-1}{q-1}$. The code $C_{(m,r)}$ is  $(r,q)$-completely transitive and, hence, completely regular with covering radius $\rho = \min\{r,m\}$.
\end{theorem}

\begin{proof}
Taking into account Corollary~\ref{coro:2.3} and Proposition~\ref{prop:1.2},
it is enough to show that the code $C_{(m,r)}$ is $(r,q)$-completely transitive. Choose any two vectors $\bv, \bw \in \F_{q^r}^{\,n}$ at the same distance from $C_{(m,r)}$, and show that there is an automorphism
$\phi \in Aut_{(r,q)}(C_{(m,r)})$ such that $\bw$ and $\phi(\bv)$ are in the same coset. That is, $\bw \in \phi(\bv) + C_{(m,r)}$ or, in other words, $\syn{\bw} = S_{\phi(\bv)}$.

From Proposition~\ref{dist} we have $\rank(\syn{\bv}) = \rank(\syn{\bw})$
and this means that there exist two nonsingular $r\times r$ and $m\times m$ matrices $A$ and $B$ over $\F_q$, respectively, such that $\syn{\bw} = A\syn{\bv}B$. Since $\pcm{m,q}$ consists of all nonzero mutually linearly independent vectors of length $m$ over $\F_q$, for any such matrix $B$,
there exists a monomial matrix $M \in Aut_{m,q}(C_{(m,r)})$ such that $\pcm{m,q}\,B = M\,\pcm{m,q}$. Hence, we obtain
$$
\bw\pcm{m,q}^T = \syn{\bw} = A\syn{\bv}B = A\bv\pcm{m,q}^TB = A\bv M\pcm{m,q}^T = \phi(\bv)\pcm{m,q}^T = S_{\phi(\bv)},
$$
where $\phi \in Aut_{m,q}(C_{(m,r)})$.
\end{proof}

To compute the intersection array of the code $C_{(m,r)}$ we need
the following well known result.

\begin{lemma}[\cite{lan}]\label{lem:2.7}
The number of different $r\times m$ matrices over $F_q$, of rank $k\leq \min\{r,m\}$ is
$$
\frac{M_q(k,r)M_q(k,m)}{M_q(k,k)},
$$
where
$$
M_q(k,t) =(q^t-1)(q^t-q)\cdots (q^t-q^{k-1})
$$
represents, for $1\leq k\leq t$,  the number of injective morphisms from a vector space of dimension $k$ to a vector space of dimension $t$, both over $\F_q$.
\end{lemma}

Now we can easily see that the number of injective morphisms from a vector space of dimension $k$ into a vector space of dimension $t$, when we fix the image of the first $k-1$ vectors in the basis of the first vector space, is given by $M'_q(k,t) = q^t-q^{k-1}$, i.e. the last factor in $M_q(k,t)$. We will say that this injective morphism has only one freedom degree.

Hence, the number of different $r\times m$ matrices over $F_q$, of rank $k\leq \min\{r,m\}$, but with only one freedom degree is:
\begin{equation}\label{gran}
\frac{M'_q(k,r)M'_q(k,m)}{M_q(1,1)}= \frac{(q^r-q^{k-1})(q^m-q^{k-1})}{(q-1)}.
\end{equation}

We need one more statement.

\begin{lemma}\label{lem:2.8}
Let $C_{(m,r)}$ be the code with parity check matrix $\pcm{m,q}$. Let $\mu_i$ be the number of cosets in $C_{(m,r)}(i)$. Then
\begin{equation}\label{ci}
\mu_i = \frac{M_q(i,r)M_q(i,m)}{M_q(i,i)}.
\end{equation}
\end{lemma}

\begin{proof}
Since Proposition~\ref{dist}, for all the vectors $\bv$ in the same $C_{(m,r)}(i)$ we have that the corresponding $r\times m$ matrix $S_{\bv}$ is of rank $i$. But, all the vectors in the same coset have the same syndrome and, hence, the number of different cosets in $C(i)$ is equal to the
number of different $r\times m$ matrices over $\F_q$ of rank $i$, which
is given by Lemma \ref{lem:2.7}.
\end{proof}

Now, it is easy to compute the intersection array for the lifted codes.

\begin{theorem}\label{theo:2.2}
Let $C_{(m,r)}$ be the code of length $n$ over the field $\F_{q^r}$ with parity check matrix $\pcm{m,q}$, where $n=\frac{q^m-1}{q-1}$.
\begin{itemize}
\item Code $C_{(m,r)}$ is a completely regular code with intersection array:
$$b_i = \frac{(q^r-q^{i})(q^m-q^{i})}{(q-1)}; \; c_i=q^{i-1}\frac{q^i-1}{q-1}.
$$
\item Codes $C_{(m,r)}$ and $C_{(r,m)}$ are, in general, different, but they have the same intersection array.
    \end{itemize}
\end{theorem}

\begin{proof}
Directly from the definition of $b_i$ and $c_i$ we have:
\begin{equation}\label{bro}
\mu_ib_i = \mu_{i+1}c_{i+1}.
\end{equation}
For each index $i\in \{0,1,\cdots,\rho\}$ we know that $a_i+b_i+c_i=(q^r-1)n$ and so, for any index, it is only necessary to compute one of the values, for instance $b_i$, since the other two values comes from the last equality and~(\ref{bro}).

We begin computing $b_0$, so the number of vectors in $C_{(m,r)}(1)$ which are at distance one from one given vector in $C_{(m,r)}$. Without losing generality (since $C_{(m,r)}$ is completely regular) we can fix vector $\bo$ in $C_{(m,r)}$ and count how many vectors there are at distance one in $C_{(m,r)}(1)$. The answer is immediately
\[
n\,(q^r-1) = \frac{(q^r - 1)(q^m - 1)}{q-1}.
\]
Indeed, from Proposition~\ref{dist}, we need to count how many cosets have syndromes with rank one or, in another way, how many $r\times m$ matrices of rank one there are. From Lemma~\ref{lem:2.7} the result is $(q^r-1)(q^m-1)/(q-1)$.

Now, we know $\mu_1$ and $b_0$ and so, from~(\ref{bro}), it is easy to compute $c_1 =\dd \frac{\mu_0b_0}{\mu_1}=1$ and $a_1=(q^r-1)n-b_1-c_1$.

In general, let $1 \leq i \leq \rho$. We use the same argumentation. The value $b_i$
is the number of different $r\times m$ matrices over $F_q$, of rank $i\leq \rho$, with exactly one freedom degree. We obtain from~(\ref{gran}), that
$$
b_i = \frac{M'_q(i+1,r)M'_q(i+1,m)}{M_q(1,1)}= \frac{(q^r-q^{i})(q^m-q^{i})}{(q-1)}\,.
$$
Now, using the expressions for $b_{i-1}$, $\mu_i$ and $\mu_{i-1}$ from Lemma \ref{lem:2.8}, we obtain
$$
c_i = \frac{\mu_{i-1}b_{i-1}}{\mu_i}= q^{i-1}\frac{(q^{i}-1)}{q-1}\,.
$$

Finally, we note that all the values $b_i, c_i$ of the above intersections arrays are symmetric for $r$ and $m$, so codes $C_{(r,m)}$ and $C_{(m,r)}$ are, in general, different, but they have the same intersection numbers $b_i, c_i$ (and different values for $a_i$ for the case $r \neq m$).
\end{proof}

It is interesting to notice that perfect codes are the only class of
nontrivial codes (i.e. $q$-ary codes with cardinality $1 < |C| < q^{n-1}$
and with minimum distance $d \geq 3$) whose codes, obtained by lifting the
ground field, are completely regular.

\begin{theorem}\label{theo:2.3}
Let $C$ be the nontrivial code of length $n$ over the field $\F_q$ with
minimum distance $d \geq 3$, with covering radius $\rho \geq 1$ and parity check matrix $H$. Let $C_r$ be the code over $\F_{q^r}$, whose
parity check matrix is $H$. Then $C_r$ is completely regular,
if and only if $C$ is a Hamming code.
\end{theorem}

\begin{proof}
Since Theorem \ref{theo:2.1}, it is enough to show that when $C$ is not a Hamming code, lifting the ground field does not give a
completely regular code $C_r$.

Take a vector $\bx$ over $\F_q$ of weight two, which is covered by some codeword $\bv \in C$ of weight $w \geq 4$ (this is possible also for any
code $C$ with $d=3$ and $n > 3$, since $C$ is not perfect). Then, from Lemma~\ref{lem:2.9}, the vector $\bx'$ of weight two, having one nonzero coordinate from $\F_q$ and the other from $\F_{q^r} \setminus \F_q$ is not covered by any codeword of minimum weight. Now the two cosets $D = C_r - \bx$
and $D' = C_r - \bx'$ (both of weight two) have different weight distributions, which is impossible for completely regular code.
\end{proof}

\medskip

\begin{remark}
Any code $C_{(m,sr)}$ contains as a subcode the code
$C_{(m,r)}$ where $s,r \geq 1$ are arbitrary natural numbers. Hence, our construction induces, for any prime power $q$, an infinite family of
nested completely regular codes:
\[
C_{(m,r)} \subset C_{(m,ar)} \subset
C_{(m,a^2\,r)} \subset \ldots
\]
This nested family of codes induces, in turn,
infinite nested families of regular and completely regular partitions (see \cite[Sec. 11.1]{Bro2}) of completely regular codes into completely regular subcodes.
\end{remark}

\medskip

Finally, we end this paper with a comment about distance regular graphs.
\begin{remark}
It is well known that any linear completely regular code $C$ implies the existence of a coset distance-regular graph. From the completely regular codes $C_{(r,m)}$ described in this paper, we obtain distance-regular graphs with classical parameters (see \cite{Bro2}) which are distance-transitive since they come
from $(r,q)$-completely transitive codes. These graphs have $v = q^{r\,m}$
vertices, diameter $\rho = \min\{r,m\}$, and intersection array given by
$$
b_i = \frac{(q^r-q^{i})(q^m-q^{i})}{(q-1)}; \; c_i=q^{i-1}\frac{q^i-1}{q-1},
$$
where $0 \leq i \leq d$.

Notice that bilinear forms graphs~\cite[Sec. 9.5]{Bro2} have the same parameters and are distance-transitive too. These graphs are uniquely
defined by their parameters (see \cite[Sec. 9.5]{Bro2}). But, we did
not find in the literature (in particular in \cite{Dels}, where  the
association schemes, formed by bilinear forms, have been introduced and
their application to coding theory have been considered) such a simple description of these graphs, as coset graphs of completely regular codes, constructed by lifting Hamming codes.
\end{remark}


\begin{thebibliography}{99}

\bibitem{Bor1}
J. Borges, J. Rif\`{a}, ``On the Nonexistence of Completely
Transitive Codes", {\em IEEE Trans. on Information Theory}, vol. 46,
no. 1, pp. 279-280, 2000.

%

\bibitem{Bro2}
A.E. Brouwer, A.M. Cohen, A. Neumaier, {\em Distance-Regular
Graphs}, Springer, Berlin, 1989.

\bibitem{Dels}
P. Delsarte, ``Bilinear forms over a finite field with applications
to coding theory," J. Combin. Th. (A), vol. 25, pp. 226 - 241, 1978.

\bibitem{giudici}
M.~Giudici, C.E.~Praeger, ``Completely Transitive Codes in Hamming Graphs", Europ. J. Combinatorics  20, pp.~647-662, 1999.


\bibitem{lan} G. Landsberg, ``Ueber eine Anzahlbestimmung und eine damit zusammenh\"{a}ngende Reihe", {\em JFRAM}, 111, pp.~87-88, 1893.

\bibitem{Neum}
A. Neumaier, ``Completely regular codes," {\em Discrete Maths.},
vol. 106/107, pp. 335-360, 1992.

\bibitem{Rif1}
J. Rif\`{a}, V.A. Zinoviev, "On new completely regular $q$-ary codes",
{\em Problems of Information Transmission}, vol. 43, No. 2, pp. 97 - 112, 2007.

\bibitem{Rif2}
J. Rif\`{a}, V.A. Zinoviev, "New completely regular $q$-ary codes,
based on Kronecker products", IEEE Transactions on Information
Theory, vol. 56, No. 1, pp. 266 - 272, 2010.

\bibitem{Sole}
P. Sol\'{e}, ``Completely Regular Codes and Completely Transitive
Codes," {\em Discrete Maths.}, vol. 81, pp. 193-201, 1990.

\bibitem{Sem3}
N.V. Semakov, V.A. Zinoviev, G.V. Zaitsev, "Class of maximal
equidistant codes," {\em Problems of Information Transmission},
vol. 5, no. 2, pp.~ 84--87, 1969.

\end{thebibliography}
\end{document}